\documentclass[copyright,creativecommons,noderivs,noncommercial]{eptcs}

\providecommand{\event}{B. Klin, P. Soboci\'nski (Eds.): \\
6th Workshop on Structural Operational Semantics (SOS'09)}
\providecommand{\volume}{19}
\providecommand{\anno}{2010}
\providecommand{\firstpage}{1}
\providecommand{\eid}{1}

\usepackage{amsmath,amssymb}
\usepackage{graphicx}
\usepackage{theorem}

\makeatletter  % so we can use it in names, like a style file.

\newcommand{\mathify}[1]{\ifmmode{#1}\else\mbox{$#1$}\fi}
\def\mathy[[#1]]{\mathify{#1}} 
\def\.#1{\mathify{#1}}

\newcommand{\ie}{{\em i.e.}}

\newcommand{\eg}{{\em e.g.}}

\let\TT=\tt

\renewcommand{\b}[1]{{\bf #1}}           % bold face

\newcommand{\e}[1]{{\em #1\/}}

\newcommand{\s}[1]{{\sf #1}}

\newcommand{\eqnref}[1]{(\ref{#1})}

\let\phi=\varphi

\def\athenasymbol
  {\hbadness=10000000 \hfuzz = 10000pt
   \hbox{\raise 4pt\hbox to -3.1pt {{$\bigtriangleup$}}
   \lower3pt\hbox{{\rm +}}
   }}

\newenvironment{proof}{{\bf Proof:}}{\mbox{}\unskip~~\hfill$\Box$\medskip}

\newcommand{\bigger}[3]{\setbox0=\hbox{$#3$}\ht0=1.05\ht0\mathify{#1\box0#2}}

\def\Biggerch#1#2#3#4#5{{
  \setbox0=\hbox{$#1#4$}\ht0=#5\ht0\dp0=#5\dp0\mathify{#2\box0#3}}
}

\def\Biggersty#1#2#3#4{
  \mathify{\mathchoice
   {\Biggerch{\displaystyle}{#1}{#2}{#3}{#4}}
   {\Biggerch{\textstyle}{#1}{#2}{#3}{#4}}
   {\Biggerch{\scriptstyle}{#1}{#2}{#3}{#4}}
   {\Biggerch{\scriptscriptstyle}{#1}{#2}{#3}{#4}}}
}

\def\.#1{\ifmmode%
\Biggersty{\left(}{\right)}{#1}{1.00}%
\else(#1)\fi}

\def\set#1{\bigger{\left\{}{\right\}}{#1}}
\def\parens#1{\bigger{\left(}{\right)}{#1}}

\newcommand{\Or}{\vee}

\newcommand{\OR}{\bigvee}

\newcommand{\Implies}{\mathrel{\Rightarrow}}

\newcommand{\Iff}{\mathify{\mathrel{\Leftrightarrow}}}

\newcommand{\Union}{\mathop{\bigcup}}
\newcommand{\union}{\mathbin{\cup}}
\newcommand{\intersect}{\mathbin{\cap}}

\renewcommand{\emptyset}{\mathify{\varnothing}}

\def\StackArray#1{\begin{array}{c} #1 \end{array}}

\def\clap#1#2{
  {
   \setbox0=\hbox{\mathify{#1}}
   \setbox1=\hbox{\mathify{#2}}
   \ifdim\wd0>\wd1
      \setbox2=\box0
      \setbox0=\box1
      \setbox1=\box2
      \fi
   \dimen0=\wd1
   \advance\dimen0 by -\wd0
   \divide\dimen0 by 2
   \dimen1=-\wd0
   \advance\dimen1 by -\dimen0
   \hskip\dimen0\box0\hskip \dimen1
   \box1
    }
  }

\newcommand{\mv}[1]{\mathrel{\stackrel{#1}{\rightarrow}}}

\newcommand{\nv}[1]{\mathrel{\stackrel{#1}{\nrightarrow}}}

\newcommand{\fix}[2]{\mathify{\r{rec} \left[#1 \Leftarrow #2\right]}}
\def\rec#1.#2{\fix{#1}{#2}}
\newcommand{\nil}{\mbox{\bf 0}}

\def\bisim{\mathrel{\underline{\leftrightarrow}}}

\newcommand{\sat}{\mathrel{\models}}

\def\mathbox#1{\hbox{$#1$}}

\def\rsap{\mathrel{\lower4pt\mathbox{\stackrel{\textstyle\sqsubset}{
            \raise1pt\mathbox{\scriptstyle\rightharpoondown}}}}}

\def\rsge{\mathrel{\lower4pt\mathbox{\stackrel{\textstyle\sqsupset}{
            \raise1pt\mathbox{\scriptstyle\leftharpoondown}}}}}

\newif\ifttordinaire\ttordinairefalse

\def\ttt{\rlap{\rm t}\kern .5ex{\rm t}}
\renewcommand{\tt}{\ifttordinaire\TT\else\ifmmode\ttt\else\TT\fi\fi}
\newcommand{\ff}{\rlap{\rm f}\kern .5ex{\rm f}}

\newcommand{\Rule}[2] 
{\[
\begin{array}{c}
\StackArray{#1} \\\hline \StackArray{#2}
\end{array}
\]
}

\newcommand{\RuleLab}[3]{
\begin{equation}
\label{#3}
\begin{array}{c}
#1 \\\hline #2
\end{array}
\end{equation}
}

\newcommand{\TwoRules}[4]{
\begin{displaymath}
\begin{array}{c}
#1 \\\hline 
#2
\end{array}
\qquad
\begin{array}{c}
#3 \\\hline 
#4
\end{array}
\end{displaymath}
}

\newcommand{\TwoRulesLab}[5]{
\begin{equation}%\begin{displaymath}
\begin{array}{c}
#1 \\\hline 
#2
\end{array}
\qquad
\begin{array}{c}
#3 \\\hline 
#4
\end{array}
\label{#5}
\end{equation}
}

\newcommand{\TwoRulesLabFullStop}[5]{
\begin{equation}%\begin{displaymath}
\begin{array}{c}
#1 \\\hline 
#2
\end{array}
\qquad
\begin{array}{c}
#3 \\\hline 
#4
\end{array}
\label{#5}
\enspace .
\end{equation}
}

\newcommand{\ThreeRules}[6]{
\begin{displaymath}
\begin{array}{c}
#1 \\\hline 
#2
\end{array}
\qquad
\begin{array}{c}
#3 \\\hline 
#4
\end{array}
\qquad
\begin{array}{c}
#5 \\\hline 
#6
\end{array}
\end{displaymath}
}

\newcommand{\ThreeRulesLab}[7]{
\begin{equation}
\begin{array}{c}
#1 \\\hline 
#2
\end{array}
\qquad
\begin{array}{c}
#3 \\\hline 
#4
\end{array}
\qquad
\begin{array}{c}
#5 \\\hline 
#6
\end{array}
\label{#7}
\end{equation}
}

\newcommand{\FourRulesFullStop}[8]{
\begin{displaymath}
\begin{array}{c}
#1 \\\hline
#2
\end{array}
\qquad
\begin{array}{c}
#3 \\\hline
#4
\end{array}
\qquad
\begin{array}{c}
#5 \\\hline
#6
\end{array}
\qquad
\begin{array}{c}
#7 \\\hline
#8
\end{array}
\enspace .
\end{displaymath}
}

\newcommand{\Act}{{\sf Act}}
\newcommand{\Var}{\s{Var}}

\def\ui/{\"\i}  % i-umlaut

\def\Quantify#1#2{\mathify{#1 #2 \b{.}\,}}
\def\Forall#1.{\Quantify{\forall}{#1}}
\def\Exists#1.{\Quantify{\exists}{#1}}
\def\Nexists#1.{\Quantify{\nexists}{#1}}

\makeatother % undo @=wierd

\newcommand{\cali}{\mathify{\cal I}}

\newcommand{\xvec}{\mathify{\vec{x}}}
\newcommand{\yvec}{\mathify{\vec{y}}}

\newcommand{\Pvec}{\mathify{\vec{P}}}
\newcommand{\Qvec}{\mathify{\vec{Q}}}

\let\oldalpha=\alpha
\let\oldbeta=\beta
\let\oldgamma=\gamma
\let\olddelta=\delta
\let\oldepsilon=\epsilon
\let\oldzeta=\zeta
\let\oldeta=\eta
\let\oldtheta=\theta
\let\oldiota=\iota
\let\oldkappa=\kappa
\let\oldmu=\mu
\let\oldnu=\nu
\let\oldpi=\pi
\let\oldxi=\xi
\let\oldvartheta=\vartheta
\let\oldlambda=\lambda
\let\oldrho=\rho
\let\oldsigma=\sigma
\let\oldtau=\tau
\let\oldupsilon=\upsilon
\let\oldphi=\phi
\let\oldchi=\chi
\let\oldpsi=\psi
\let\oldomega=\omega
\let\oldGamma=\Gamma
\let\oldDelta=\Delta
\let\oldTheta=\Theta
\let\oldLambda=\Lambda
\let\oldXi=\Xi
\let\oldPi=\Pi
\let\oldSigma=\Sigma
\let\oldUpsilon=\Upsilon
\let\oldPhi=\Phi
\let\oldPsi=\Psi
\let\oldOmega=\Omega

\def\alpha{\mathify{\oldalpha}}
\def\beta{\mathify{\oldbeta}}
\def\gamma{\mathify{\oldgamma}}
\def\delta{\mathify{\olddelta}}
\def\epsilon{\mathify{\oldepsilon}}
\def\zeta{\mathify{\oldzeta}}
\def\eta{\mathify{\oldeta}}
\def\theta{\mathify{\oldtheta}}
\def\iota{\mathify{\oldiota}}
\def\kappa{\mathify{\oldkappa}}
\def\mu{\mathify{\oldmu}}
\def\nu{\mathify{\oldnu}}
\def\pi{\mathify{\oldpi}}
\def\xi{\mathify{\oldxi}}
\def\vartheta{\mathify{\oldvartheta}}
\def\lambda{\mathify{\oldlambda}}
\def\rho{\mathify{\oldrho}}
\def\sigma{\mathify{\oldsigma}}
\def\tau{\mathify{\oldtau}}
\def\upsilon{\mathify{\oldupsilon}}
\def\phi{\mathify{\oldphi}}
\def\chi{\mathify{\oldchi}}
\def\psi{\mathify{\oldpsi}}
\def\omega{\mathify{\oldomega}}
\def\Gamma{\mathify{\oldGamma}}
\def\Delta{\mathify{\oldDelta}}
\def\Theta{\mathify{\oldTheta}}
\def\Lambda{\mathify{\oldLambda}}
\def\Xi{\mathify{\oldXi}}
\def\Pi{\mathify{\oldPi}}
\def\Sigma{\mathify{\oldSigma}}
\def\Upsilon{\mathify{\oldUpsilon}}
\def\Phi{\mathify{\oldPhi}}
\def\Psi{\mathify{\oldPsi}}
\def\Omega{\mathify{\oldOmega}}

\newtheorem{defi}{Definition}[section]
\newtheorem{theo}{Theorem}[section]
\newtheorem{prop}{Proposition}[section]
\newtheorem{lemm}{Lemma}[section]
\newtheorem{coro}{Corollary}[section]
\newtheorem{exam}{Example}[section]
\newtheorem{fatto}{Fact}[section]
\newtheorem{rmk}{Remark}[section]
\newtheorem{clm}{Claim}
\newtheorem{asz}{Simplifying Assumption}

\newenvironment{remark}{\begin{rmk} \rm  }{\end{rmk}}

\newenvironment{definition}{\begin{defi} \rm  }{\end{defi}}
\newenvironment{theorem}{\begin{theo} \rm  }{\end{theo}}

\newenvironment{lemma}{\begin{lemm} \rm  }{\end{lemm}}

\newenvironment{example}{\begin{exam} \rm  }{\end{exam}}

\def\imp{\Rightarrow}

\def\ante#1{\mathy[[\s{ante}\parens{#1}]]}
\def\cons#1{\mathy[[\s{cons}\parens{#1}]]}

\def\vars#1{\mathy[[\s{vars}\parens{#1}]]}
\def\ruloids#1{\mathy[[R_G\parens{#1}]]}

\def\hyps#1{\s{hyps}( #1 )}
\def\interleave{\mathbin{\clap{|}{\|}}}

\newcommand{\diam}[1]{\langle #1 \rangle}
\newcommand{\support}[1]{\mbox{\sf Supp}(#1)}
\newcommand{\triv}{\mbox{TRIV}}

\newcommand{\clock}{\Omega_\Act}

\newcommand{\deq}{\mathrel{\stackrel{\scriptscriptstyle\Delta}{=}}}
\newcommand{\mathdef}[1]{\relax\ifmmode #1\else $#1$\fi}
\newcommand{\proofrule}[2]{\setlength{\arrayrulewidth}{.6pt}%
  {\mathdef{\begin{array}{@{}c@{}}%
	 \begin{array}{@{}l@{}} #1\raisebox{-.1em}{\strut}\end{array}\\
		\hline \raisebox{.1em}{\strut}#2\end{array}}}}

\renewcommand{\bisim}{\mathrel{\underline{\leftrightarrow}}}

\newcommand{\OpenT}{\setlength{\unitlength}{1ex}
		\begin{picture}(1.25,1.45)
	       \put(0.4,0){\line(0,1){1.45}}
	  \put(0.85,0){\line(0,1){1.45}}
    \put(-0.1,1.5){\line(1,0){1.45}}
	\put(0,1.45){\line(1,0){1.25}}
	   \put(0.4,0){\line(1,0){0.45}}
	      \end{picture}}
\newcommand{\terms}[1]{\OpenT ( #1 )}
\newcommand{\closed}[1]{{\rm T}( #1 )}
\def\syneq{\equiv}
\newcommand{\Algbis}[1]{\s{Bisim}\parens{#1}}
\newcommand{\Classbis}[1]{\s{BISIM}\parens{#1}}

\newcommand{\subst}[2]{#1 #2}
\newcommand{\init}[1]{\mbox{\em init}(#1)}

\def\obisim{\bisim^{\scriptstyle {RM}}}
\def\sourcevars#1{\mathy[[\s{SV}(#1)]]}
\def\targetvars#1{\mathy[[\s{TV}(#1)]]}

\newcommand{\T}{\mbox{\sf True}}
\newcommand{\F}{\mbox{\sf False}}

\title{A Bisimulation-based Method for Proving the Validity of
Equations in GSOS Languages\thanks{The work of the authors has been
partially supported by the projects ``The Equational Logic of Parallel
Processes'' (nr.~060013021) and ``New Developments in Operational
Semantics'' (nr.~080039021) of the Icelandic Research Fund. The first
author dedicates this article to the memory of his mother, Imelde
Diomede Aceto, who passed away on June 26, 2008.}}

\author{Luca Aceto \quad
Matteo Cimini \quad
Anna Ingolfsdottir
\institute{School of Computer Science, Reykjavik University,
\\ Kringlan 1, IS-103 Reykjavik, Iceland} 
\email{$\{$luca, matteo, annai$\}$@ru.is}
}

\def\authorrunning{Aceto, Cimini and Ingolfsdottir}
\def\titlerunning{A Bisimulation-based Method for Proving the Validity of
Equations in GSOS Languages}

\begin{document}
\maketitle

\begin{abstract}
This paper presents a bisimulation-based method for establishing the
soundness of equations between terms constructed using operations
whose semantics is specified by rules in the GSOS format of Bloom,
Istrail and Meyer. The method is inspired by de Simone's
FH-bisimilarity and uses transition rules as schematic transitions in
a bisimulation-like relation between open terms. The soundness of the
method is proven and examples showing its applicability are
provided. The proposed bisimulation-based proof method is incomplete,
but the article offers some completeness results for restricted
classes of GSOS specifications.
\end{abstract}

\section{Introduction}\label{Sect:intro}

Equations play a fundamental role in the development of the theory and
practice of process calculi and programming languages since they offer
a mathematically appealing and concise way of stating the `laws of
programming' (to borrow the title of a paper by Hoare et
al.~\cite{HH87}) that apply to the language at hand. In the setting of
process calculi, the study of equational axiomatizations of
behavioural relations has been a classic area of investigation since,
\eg, the early work of Hennessy and Milner~\cite{HM85,Mi84}, who
offered complete axiom systems for bisimilarity~\cite{Pa81} over the
finite and regular fragments of Milner's CCS~\cite{Mi89}. Such
axiomatizations capture the essence of bisimilarity over those
fragments of CCS in a syntactic, and often revealing, way and
potentially pave the way for the verification of equivalences between
processes by means of theorem proving techniques. Despite these early
achievements, the search for axiomatizations of process equivalences
that are powerful enough to establish all the valid equations between
open process terms---that is, terms possibly containing
variables---has proven to be a very difficult research problem;
see~\cite{AcetoFIL05} for a survey of results in this area. For
instance, to the best of our knowledge, there is no known
axiomatization of bisimilarity over recursion-free CCS that is
complete over open terms. Stepping stones towards such a result are
offered in, \eg,~\cite{AcetoFIL09,AcetoILT08}.

The most basic property of any equation is that it be {\em sound} with
respect to the chosen notion of semantics. Soundness proofs are often
lengthy, work-intensive and need to be carried out for many equations
and languages. It is therefore not surprising that the development of
general methods for proving equivalences between open terms in
expressive process calculi has received some attention since the early
developments of the algebraic theory of processes---see, \eg, the
references~\cite{BruniFMM00,LarsenX1991,Rensink00,dS85,Weerdenburg2008}
for some of the work in this area over a period of over 20 years. This
article offers a contribution to this line of research by developing a
bisimulation-based method, which we call {\em rule-matching
bisimilarity}, for establishing the soundness of equations between
terms constructed using operations whose semantics is specified by
rules in the GSOS format of Bloom, Istrail and
Meyer~\cite{BloomIM95}. Rule-matching bisimilarity is inspired by de
Simone's FH-bisimilarity~\cite{dS85} and uses transition rules as
transition schemas in a bisimulation-like relation between open
terms. We prove that rule-matching bisimilarity is a sound proof
method for showing the validity of equations with respect to
bisimilarity and exhibit examples witnessing its incompleteness.

The incompleteness of rule-matching bisimilarity is not unexpected and
raises the question whether the method is powerful enough to prove the
soundness of `interesting' equations. In order to offer a partial
answer to this question, we provide examples showing the applicability
of our proof method. In particular, our method does not only apply to
a more expressive rule format than the one proposed by de Simone
in~\cite{dS85}, but is also a sharpening of de Simone's
FH-bisimilarity over de Simone languages. See
Section~\ref{Sect:examples}, where we apply rule-matching bisimilarity
to prove the soundness of the equations in de Simone's `clock
example'. (This example was discussed by de Simone in~\cite{dS85} to
highlight the incompleteness of FH-bisimilarity.) On the theoretical
side, we also offer some completeness results for restricted classes
of GSOS specifications.

Overall, we believe that, while our conditions are neither necessary
nor in general can they be checked algorithmically, they frequently
hold, and they are more accessible to machine support than a direct
proof of soundness.

The paper is organized as follows. Sections~\ref{Sect:preliminaries}
and~\ref{Sect:ruloids} introduce the necessary preliminaries on the
GSOS rule format that are needed in the reminder of the paper. In
particular, Section~\ref{Sect:ruloids} recalls the notion of ruloid,
which plays a key role in the technical developments to follow. In
Section~\ref{transition-logic}, we introduce a simple logic of
transition formulae and establish a decidability result for the
validity of implications between formulae. Implication between certain
kinds of transition formulae that are naturally associated with the
premises of (sets of) ruloids is used in the definition of
rule-matching bisimilarity in Section~\ref{Sect:rule-match}. In that
section, we prove that rule-matching bisimilarity is a sound method
for showing the validity of equations in GSOS languages modulo
bisimilarity and exhibit examples witnessing its incompleteness. We
apply rule-matching bisimilarity to show the validity of some sample
equations from the literature on process algebra in
Section~\ref{Sect:examples}. We then offer some partial completeness
results for rule-matching bisimilarity
(Section~\ref{Sect:partial-completeness}). The paper concludes with a
discussion of related and future work
(Section~\ref{Sect:related-work}).

\section{Preliminaries}
\label{Sect:preliminaries}

We assume familiarity with the basic notation of process algebra and
structural operational semantics; see \eg{}
\cite{AcetoFV2001,BW90,BloomIM95,GrV92,He88a,Ho85,Mi89,MousaviRG07,Plotkin04a}
for more details.

Let $\Var$ be a countably infinite set of {\em process variables} with
typical elements $x,y$.  A {\em signature} $\Sigma$ consists of a set
of {\em operation symbols}, disjoint from $\Var$, together with a
function {\em arity} that assigns a natural number to each operation
symbol.  The set $\terms{\Sigma}$ of {\em terms} built from the
operations in $\Sigma$ and the variables in $\Var$ is defined in the
standard way.  We use $P,Q,\ldots$ to range over terms and the symbol
$\syneq$ for the relation of syntactic equality on terms.  We denote
by $\closed{\Sigma}$ the set of \e{closed} terms over $\Sigma$, \ie,
terms that do not contain variables, and will use $p,q,\ldots$ to
range over it. An operation symbol $f$ of arity $0$ will be often
called a \e{constant} symbol, and the term $f()$ will be abbreviated
as $f$.

Besides terms we have \e{actions}, elements of some given nonempty, finite
set {$\Act$}, which is ranged over by  {$a,b,c,d$}.
A \e{positive transition formula} is a triple of two terms and an action,
written {$P \mv{a} P'$}.  A \e{negative transition formula} is a pair of a
term and an action, written {$P \nv{a}$}.

A {\em (closed) $\Sigma$-substitution} is a function $\sigma$ from
variables to (closed) terms over the signature $\Sigma$.  For $t$ a
term or a transition formula, we write $\subst{t}{\sigma}$ for the
result of substituting $\sigma (x)$ for each $x$ occurring in $t$, and
$\vars{t}$ for the set of variables occurring in $t$. A
\e{{$\Sigma$}-context} {$C[\xvec]$} is a term in which at most the
variables {$\xvec$} appear.  {$C[\Pvec]$} is {$C[\xvec]$} with
{$x_{i}$} replaced by {$P_{i}$} wherever it occurs.
 
\begin{definition}[GSOS Rule]
\label{gsos-rule-defn}
Suppose $\Sigma$ is a signature.
A \e{GSOS rule} {$\rho$} over $\Sigma$ is a rule of the form:
\RuleLab{
  \Union_{i=1}^{l} \set{x_{i} \mv{a_{ij}} y_{ij} | 1 \leq j \leq m_{i}}
  \setbox0=\mathbox{\set{x_{i} \nv{b_{ik}} | 1 \leq k \leq n_{i}}}
  % I had some trouble getting lines to stay apart here...
  \dp0=1.4\dp0
  ~~ \union ~~ \Union_{i=1}^{l} \box0
}{
  f(x_{1},\ldots,x_{l}) \mv{c} C[\vec{x}, \vec{y}]
}
{general-gsos-rule}
where all the variables are distinct, {$m_{i}, n_{i} \geq 0$}, {$a_{ij}$}, 
{$b_{ik}$}, and {$c$} are actions, 
$f$ is an operation symbol from $\Sigma$ with arity $l$,
and {$C[\vec{x}, \vec{y}]$} is a $\Sigma$-context.

It is useful to name components of rules.
The operation symbol {$f$} is the \e{principal operation} of the rule,
and the term {$f(\xvec)$} is the \e{source}.
{$C[\xvec,\yvec]$} is the \e{target}; {$c$} is the \e{action};
the formulae above the line are the \e{antecedents} (sometimes denoted by 
$\ante{\rho}$); and the formula below the line is the \e{consequent} (sometimes
denoted by {$\cons{\rho}$}).

For a GSOS rule {$\rho$}, \sourcevars{\rho} and \targetvars{\rho} are
the sets of source and target variables of {$\rho$}; that is,
{$\sourcevars{\rho}$} is the set of variables in the source of {$\rho$}, and
{$\targetvars{\rho}$} is the set of {$y$}'s for antecedents {$x \mv{a} y$}. 
\end{definition}
%%
%%All rules in this paper (and almost all rules appearing in the literature
%%on process algebra) are examples of GSOS rules.
\begin{definition}
\label{gsos-system}
A \e{GSOS language} is a pair $G = ( \Sigma_G , R_G)$ where $\Sigma_G$ is a
finite signature and $R_G$ is a finite set of GSOS rules over $\Sigma_G$.
\end{definition}
Informally, the intent of a GSOS rule is as follows.  Suppose that we
are wondering whether {$f(\Pvec)$} is capable of taking a {$c$}-step.
We look at each rule with principal operation {$f$} and action {$c$}
in turn.  We inspect each positive antecedent {$x_{i} \mv{a_{ij}}
y_{ij}$}, checking if {$P_{i}$} is capable of taking an
{$a_{ij}$}-step for each {$j$} and if so calling the
{$a_{ij}$}-children {{$Q_{ij}$}.} We also check the negative
antecedents; if {$P_{i}$} is \e{in}capable of taking a {$b_{ik}$}-step
for each {$k$}.  If so, then the rule \e{fires} and {$f(\Pvec) \mv{c}
C[\Pvec,\Qvec]$}.  This means that the transition relation $\mv{}_G$
associated with a GSOS language $G$ is the one defined by the rules
using structural induction over closed $\Sigma_G$-terms. This
transition relation is the unique sound and supported transition
relation. Here {\em sound} means that whenever a closed substitution
$\sigma$ `satisfies' the antecedents of a rule of the form
(\ref{general-gsos-rule}), written $\rightarrow_G, \sigma \models 
\ante{\rho}$, then $f(x_1,\ldots,x_l)\sigma \mv{c}_G
C[\vec{x}, \vec{y}]\sigma$. On the other hand, {\em supported} means
that any transition $p \mv{c}_G q$ can be obtained by instantiating
the conclusion of a rule $\rho$ of the form (\ref{general-gsos-rule})
with a substitution that satisfies its premises. In that case, we say
that $p \mv{c}_G q$ is supported by $\rho$. A rule $\rho$ is {\em
junk} in $G$ if it does not support any transition in $\rightarrow_G$.
We refer the interested
reader to~\cite{BloomIM95} for the precise definition of $\mv{}_G$ and
much more information on GSOS languages.  

For each closed term $p$, we define $\init{p} = \{ a\in\Act \mid
\exists q:~ p \mv{a}_G q\}$. For a GSOS language $G$, we let
$\init{\closed{\Sigma_G}} = \{ \init{p} \mid p\in\closed{\Sigma_G}
\}$.

The basic notion of equivalence among terms of a GSOS language we will
consider in this paper is  {\em bisimulation
equivalence}~\cite{Mi89,Pa81}.

\begin{definition}
Suppose $G$ is a GSOS language.  A binary relation
$\sim{} \subseteq \closed{\Sigma_G}\times\closed{\Sigma_G}$
over closed terms is a {\em bisimulation} if it is symmetric and 
{$p \sim q$} implies, for all $a \in\Act$,
\begin{description}
\item
If $p \mv{a}_G p'$ then, for some $q'$, $q \mv{a}_G q'$ and $p' \sim q'$.
%\item If $Q \mv{a}_G Q'$ then, for some $P'$, $P \mv{a}_G P'$ and $P' \sim Q'$.
\end{description}
We write $p \bisim_G q$ if there exists a bisimulation {$\sim$}
relating $p$ and $q$.  The subscript {$G$} is omitted when it is clear
from the context.
\end{definition}
It is well known that $\bisim_G$ is a 
congruence for all operation symbols $f$ of $G$~\cite{BloomIM95}. 

Let $\Algbis{G}$ denote the quotient
algebra of closed {$\Sigma_G$}-terms modulo bisimulation. 
Then, for
$P, Q \in\terms{\Sigma_G}$,
\[
\Algbis{G} \models P = Q ~~~ \Iff ~~~
( \forall \mbox{ closed } \Sigma_G\mbox{-substitutions } \sigma :
\subst{P}{\sigma} \bisim_{G} \subst{Q}{\sigma} ).
\]
In what follows, we shall sometimes consider equations that hold over
all GSOS languages that extend a GSOS language $G$ with new operation
symbols and rules for the new operations. The following notions
from~\cite{ABV94} put these extensions on a formal footing.
\begin{definition}\label{Sect:disjoint-extensions}
A GSOS language $G'$ is a {\em disjoint extension} of a GSOS language
$G$ 
%%, notation $G \cextends H$, 
if the signature and rules of {$G'$}
include those of {$G$}, and {$G'$} introduces no new rules for
operations of {$G$}.
\end{definition}
If $G'$ disjointly extends $G$ then $G'$ introduces no new outgoing
transitions for the closed terms of $G$.  This means in particular
that $P \bisim_G Q$ iff $P \bisim_{G'} Q$, for $P, Q
\in\closed{\Sigma_G}$. (More general conservative extension results
are discussed in, \eg,~\cite{Fokkink98,Mousavi05-ICALP}.)

For $G$ a GSOS language, let $\Classbis{G}$ stand for the class of all
algebras $\Algbis{G'}$, for $G'$ a disjoint extension of $G$.  Thus we
have, for $P, Q \in\terms{\Sigma_G}$,
\[
\Classbis{G} \models P = Q ~~~ \Iff ~~~ ( \forall G' : %%G \cextends G'
G' \text{  a disjoint extension of } G
~ \implies ~ \Algbis{G'} \models P = Q ).
\]
Checking the validity of a statement of the form $\Algbis{G} \models P
= Q$ or $\Classbis{G} \models P = Q$ according to the above definition
is at best very impractical, as it involves establishing bisimilarity
of all closed instantiations of the terms $P$ and $Q$.  It would thus
be helpful to have techniques that use only information obtainable
from these terms and that can be used to this end.  The development of
one such technique will be the subject of the remainder of this paper.

\paragraph{Eliminating Junk Rules} 
\label{junk-elimination}

Note that the definition of a GSOS language given above does not
exclude {\em junk} rules, \ie, rules that support no transition in
$\rightarrow_G$.  For example, the rule \Rule{x \mv{a} y ,\qquad x
\nv{a}}{f(x) \mv{a} f(y)} has contradictory antecedents and can never
fire.  Also it can be the case that a (seemingly innocuous) rule like
\Rule{x \mv{a} y}{f(x) \mv{b} f(y)} does not support any transition if
$\rightarrow_G$ contains no $a$-transitions. The possible presence of
junk rules does not create any problems in the development of the
theory of GSOS languages as presented in \cite{ABV94,BloomIM95}
and the authors of those papers saw no reason to deal with these rules
explicitly.

Our aim in this paper is to develop a test for the validity of
equalities between open terms in GSOS languages. The test we shall
present in later sections is based upon the idea of using GSOS rules
as `abstract transitions' in a bisimulation-like equivalence between
open terms. In order to ease the applicability of this method, it is
thus desirable, albeit not strictly necessary, to eliminate junk rules
from GSOS languages, as these rules would be interpreted as `potential
transitions' from a term which, however, cannot be realized.

Consider, for example, the trivial GSOS language $\triv$ with unary operations 
$f$ and $g$, and rule 
$$ f(x) \mv{a} f(x) \enspace . $$
It is immediate to see that $\Algbis{\triv} \models f(x) = g(y)$ as the set of
closed terms in $\triv$ is empty. However, if we considered the rule for 
$f$ as a transition from $f(x)$ in a simple-minded way, we would be led to 
distinguish $f(x)$ and $g(y)$ as the former has a transition while the latter 
does not. Obviously, the rule for $f$ given above is junk.

Clearly junk rules can be removed from a GSOS language $G$ without
altering the associated transition relation.  
%%TO BE REINSTATED IN THE FULL VERSION
%%Note, moreover, that it
%%is legitimate to eliminate all the junk rules in $G$ at once, so to
%%speak. This is because whenever $\rho_1$ and $\rho_2$ are junk rules
%%in $G$, then $\rho_2$ is still junk in the GSOS language obtained from
%%$G$ by removing $\rho_1$, as the transition relation is the same in
%%the two GSOS languages.
%% END OF TEXT TO BE REINSTATED
Of course, in order to be able to remove junk rules from a GSOS
language, we need to be able to discover effectively what rules are
junk. This is indeed possible, as the following theorem, due to Aceto,
Bloom and Vaandrager~[Theorem~5.22]\cite{AcetoFV2001}, shows.

\begin{theorem}\label{decidability-of-junkness}
Let $G=(\Sigma_G,R_G)$ be a GSOS language. Suppose that $\rho\in
R_G$. Then it is decidable whether $\rho$ is junk in $G$.
\end{theorem}
As a consequence of the above theorem, all the junk rules in a GSOS
language can be effectively removed in a pre-processing step before
applying the techniques described in the subsequent sections. Thus we
will henceforth restrict ourselves to GSOS languages without junk
rules.

\section{Ruloids and the Operational Specification of Contexts}
\label{Sect:ruloids}

As mentioned above, the essence of our method for checking the
validity of equations in GSOS languages is to devise a variation on
bisimulation equivalence between contexts which considers GSOS rules
as transitions. For primitive operations in a GSOS language $G$, the
rules in $R_G$ will be viewed as abstract transitions from terms of
the form $f(\xvec)$. However, in general, we will be dealing with
complex contexts in $\terms{\Sigma_G}$. In order to apply our ideas to
general open terms, we will thus need to associate with arbitrary
contexts a set of derived rules (referred to as {\em ruloids}
\cite{BloomIM95}) describing their behaviour.

A {\em ruloid} for a context $D[\xvec]$, with $\xvec=(x_1,\ldots,x_l)$, takes 
the form:
\RuleLab{
  \Union_{i=1}^{l} \set{x_{i} \mv{a_{ij}} y_{ij} | 1 \leq j \leq m_{i}}
    \setbox0=\mathbox{\set{x_{i} \nv{b_{ik}} | 1 \leq k \leq n_{i}}}
      % I had some trouble getting lines to stay apart here...
	\dp0=1.4\dp0
	  ~~ \union ~~ \Union_{i=1}^{l} \box0
	  }{
	    D[\xvec] \mv{c} C[\vec{x}, \vec{y}]
	    }
	    {general-ruloid}
where the variables are distinct, {$m_{i}, n_{i} \geq 0$}, {$a_{ij}$}, 
{$b_{ik}$}, and {$c$} are actions, 
and {$C[\vec{x}, \vec{y}]$} is a $\Sigma$-context.

\begin{definition}
A set of ruloids $R$ is {\em supporting}\footnote{Our terminology departs
slightly from that of \cite{BloomIM95}. Bloom, Istrail and Meyer use
`specifically witnessing' in lieu of `supporting'.} for a context $D[\xvec]$
and action $c$ iff all the consequents of ruloids in $R$ are of the form 
$D[\xvec] \mv{c} C[\vec{x}, \vec{y}]$, and whenever $D[\Pvec] \mv{c} p$, there 
are a ruloid $\rho\in R$ and a closed substitution $\sigma$ such that 
{$\cons{\rho}\sigma = D[\Pvec] \mv{c} p$} and $\rightarrow_G, \sigma \models 
\ante{\rho}$.
\end{definition}
The following theorem is a slightly sharpened version of the Ruloid Theorem in
\cite{BloomIM95}.

\begin{theorem}[Ruloid Theorem]\label{ruloid-thm}
Let $G$ be a GSOS language and $X \subseteq \Var$ be a finite set of variables.
For each $D[\xvec]\in\terms{\Sigma_G}$ and action $c$, there exists a finite
set $R_{D,c}$ of ruloids of the form (\ref{general-ruloid}) such that:
\begin{enumerate}
\item the ruloids in $R_{D,c}$ are sound and supporting for $D[\xvec]$, and 
\item for every $\rho\in R_{D,c}$, $\targetvars{\rho} \cap X = \emptyset$.
\end{enumerate}
\end{theorem}
\begin{proof}
A straightforward adaptation of the proof of the corresponding result
in \cite{BloomIM95}, where we take care in choosing the target
variables in ruloids so that condition 2 in the statement of the
theorem is met.
\end{proof}

\begin{definition}
Let $G$ be a GSOS language. For each $D[\xvec]\in\terms{\Sigma_G}$, the ruloid
set of $D[\xvec]$, notation $\ruloids{D[\xvec]}$, is the union of the sets
$R_{D,c}$ given by Theorem~\ref{ruloid-thm}.
\end{definition}
The import of the Ruloid Theorem is that the operational semantics of
an open term $P$ can be described by a finite set $\ruloids{P}$ of
derived GSOS-like rules. Examples of versions of the above result for
more expressive formats of operational rules may be found in, \eg, the
references~\cite{Bloom04,FokkinkGW06}.

\begin{example}
Consider a GSOS language $G$ containing the sequencing operation
specified by the following rules (one such pair of rules for each
$a\in \Act$).

\TwoRulesLab {x \mv{a} z} {x;y \mv{a} z;y}
	     {x \nv{b}~(\forall b\in\Act), y \mv{a} z}{x;y \mv{a} z}
             {rules4sequencing}
Let {$R[x,y,z]= x;(y;z)$} and {$L[x,y,z]=(x;y);z$}. The ruloids for {$L$} and 
{$R$} are:

\ThreeRulesLab {x \mv{a} x'} {L \mv{a} (x';y);z \\ R \mv{a} x';(y;z)}
               {x \nv{}, y \mv{a} y'} {L \mv{a} y';z \\ R \mv{a} y';z}
	       {x \nv{}, y \nv{}, z \mv{a} z'} {L \mv{a} z' \\ R \mv{a} z'}
               {ruloids4sequencing}
where we write $x \nv{}$ in the antecedents of ruloids as a shorthand for {$x
\nv{b}~(\forall b\in\Act)$}.
\end{example}
\begin{remark}\label{Rem:junk-ruloids}
Note that the set $\ruloids{D[\xvec]}$ of ruloids for a context
$D[\xvec]$ in a GSOS language $G$ may contain junk ruloids even when
$G$ has no junk rule. For example, consider the GSOS language
with constants $a$ and $\nil$, unary operation $g$ and binary
operation $f$ with the following rules. 
\ThreeRules {} {a \mv{a} \nil}
            {x \mv{a} x', y \mv{b} y'} {f(x,y) \mv{a} \nil}
	    {x \nv{a}} {g(x) \mv{b} \nil}
None of the above rules is junk. However, the only ruloid for the
context $f(x,g(x))$ is
\[
\frac{x \mv{a} x', x \nv{a}}{f(x,g(x)) \mv{a} \nil} \enspace ,
\]
which is junk. However, junk ruloids can be removed from the set of
ruloids for a context using Theorem~\ref{decidability-of-junkness}. In
what follows, we shall assume that the set of ruloids we consider
have no junk ruloids.
\end{remark}
In the standard theory on GSOS, it was not necessary to pay much attention to 
the variables in rules and ruloids, as one was only interested in the 
transition relation they induced over {\em closed terms}. (In the terminology 
of \cite{GrV92}, all the variables occurring in a GSOS rule/ruloid are not 
{\em free}.) Here, however, we intend to use ruloids as abstract transitions 
between open terms. In this framework it becomes desirable to give a more 
reasoned account of the role played by variables in ruloids, as the following example shows.

\begin{example}
Consider a GSOS language $G$ containing the unary operations $f$ and 
$g$ with the following rules. 

\TwoRules{x \mv{a} y}{f(x) \mv{a} y}
         {x \mv{a} z}{g(x) \mv{a} z}
It is easy to see that 
$\Algbis{G} \sat f(x) = g(x)$, regardless of the precise description of $G$. However, in order to prove this equality, 
any bisimulation-like equivalence relating open terms in $\terms{\Sigma_G}$ 
would have to relate the variables $y$ and $z$ in some way. Of course, this will
have to be done carefully, as $y$ and $z$ are obviously not equivalent in any 
nontrivial language.
\end{example}
As the above-given example shows, in order to be able to prove many simple 
equalities between open terms, it is necessary to develop techniques which
allow us to deal with the target variables in ruloids in a reasonable way. In 
particular, we should not give too much importance to the names of target 
variables in ruloids.

\begin{definition}[Valid Ruloids]\label{valid-ruloids}
Let $G$ be a GSOS language and $P\in \terms{\Sigma_G}$. We say that a ruloid
$\rho = \proofrule{H}{P\mv{a} P'}$ is {\em valid} for $P$ iff there exist 
$\rho'\in\ruloids{P}$ and an injective map $\sigma: \targetvars{\rho'}
\rightarrow (\Var - \sourcevars{\rho})$ such that $\rho$ is identical to 
$\rho'\sigma$. 
\end{definition}
For example, it is immediate to notice that the rules 
\TwoRules{x \mv{a} z}{f(x) \mv{a} z}
         {x \mv{a} y}{g(x) \mv{a} y}
are valid for the contexts $f(x)$ and $g(x)$ in the above-given example. 
Note, moreover, that each ruloid in $\ruloids{P}$ is 
a valid ruloid for $P$.

The following lemma states that, if $\rho'$ is obtained from $\rho$ as in
Definition~\ref{valid-ruloids}, then $\rho$ and $\rho'$ are, in a sense,
semantically equivalent ruloids.

\begin{lemma}\label{properties-of-valid-ruloids}
Let $G = (\Sigma_G,R_G)$ be a GSOS language and $P\in\terms{\Sigma_G}$. 
Assume that $\rho$ is a valid ruloid for $P$ because $\rho = \rho' \sigma$ for
some $\rho'\in \ruloids{P}$ and injective $\sigma:\targetvars{\rho'}\rightarrow
\Var - \sourcevars{\rho}$. Then:
\begin{enumerate}
\item $\rho$ is sound for $\rightarrow_G$;
\item $\support{\rho} = \support{\rho'}$, where, for a GSOS rule/ruloid
$\hat{\rho}$, $\support{\hat{\rho}}$ denotes the set of transitions supported
by $\hat{\rho}$.
\end{enumerate}
\end{lemma}
The set of valid ruloids for a context $P$ is infinite. However, by
Theorem~\ref{ruloid-thm}, we can always select a finite set of valid
ruloids for $P$ which is sound and supporting for it. We will often
make use of this observation in what follows.  

\section{A Logic of Transition Formulae}\label{transition-logic}

The set of ruloids associated with an open term $P$ in a GSOS language
characterizes its behaviour in much the same way as GSOS rules give the
behaviour of GSOS operations. In fact, by Theorem~\ref{ruloid-thm}, every
transition from a closed term of the form $P\sigma$ can be inferred from a
ruloid in $\ruloids{P}$. 

The antecedents of ruloids give the precise conditions under which
ruloids fire.  When matching ruloids in the definition of the
bisimulation-like relation between open terms that we aim at defining,
we will let a ruloid $\rho$ be matched by a set of ruloids $J$ only if
the antecedents of $\rho$ are stronger than those of the ruloids in $J$,
{\ie}, if whenever $\rho$ can fire under a substitution $\sigma$, then
at least one of the ruloids in $J$ can. In order to formalize this idea,
we will make use of a simple propositional logic of initial transition
formulae.

We define the language of \e{initial transition formulae} to be 
propositional logic with propositions of the form {$x \mv{a}$}. Formally, the
formulae of such a logic are given by the following grammar:
\[
F :: = \quad \T \mid x \mv{a}~\mid \neg F \mid F \wedge F 
\enspace .
\]
As usual, we write $\F$ for $\neg \T$, and $F \Or F'$ for $\neg ( \neg F \wedge
\neg F')$. 

Let $G$ be a GSOS language. A $G$-model for initial transition formulae is a 
substitution {$\sigma$} of processes (closed $\Sigma_G$-terms) for variables. 
We write $\rightarrow_G,\sigma \sat F$ if the closed substitution $\sigma$ is 
a model of the initial transition formula $F$. The satisfaction relation $\sat$ is defined by 
structural recursion on $F$ in the obvious way. In particular, 
\[
\rightarrow_G,\sigma \sat x \mv{a} \text{ iff } \sigma(x) \mv{a}_G p,~
\text{ for some } p \enspace .
\]
The reader familiar with Hennessy-Milner logic~\cite{HM85} will have
noticed that the propositions of the form {$x \mv{a}$} correspond to
Hennessy-Milner formulae of the form $\diam{a} \T$.  If {$H$} is a set
of positive or negative transition formulae (\eg, the hypotheses of a
rule or ruloid), then {$\hyps{H}$} is the conjunction of the
corresponding initial transition formulae.  For example,
$\hyps{\set{x\mv{a}y, z\nv{b}}} = (x \mv{a}) \wedge \neg (z \mv{b})$.
If {$J$} is a finite set of ruloids, we overload {$\hyps{\cdot}$} and
write:
\begin{equation}
\hyps{J} \deq \OR _{\rho' \in J} \hyps{\ante{\rho'}} \enspace .
\end{equation}
The semantic entailment preorder between initial transition formulae
may be now defined in the standard way; for formulae $F,F'$, we have 
$\sat_G F \imp F'$ iff every substitution that satisfies $F$ must also 
satisfy $F'$. 

In the remainder of this paper, we will use the semantic entailment
preorder between transition formulae in our test for equivalence of
open terms to characterize the fact that if one ruloid may fire, then
some other may do so too. Of course, in order to be able to use the
entailment preorder between transition formulae in our test for open
equalities, we need to able to check effectively when $\sat_G F \imp
F'$ holds. Fortunately, the semantic entailment preorder between
formulae is decidable, as the following theorem shows.

\begin{theorem}\label{decidability-of-refinement}
Let $G$ be a GSOS language. Then for all formulae $F$ and $F'$, it is decidable 
whether $\sat_G F \imp F'$ holds. 
\end{theorem}
Theorem~\ref{decidability-of-refinement} tells us that we can safely 
use semantic entailment between formulae in our simple propositional language 
in the test for the validity of open equations in GSOS languages which we will 
present in what follows.

\section{Rule-matching Bisimulation}\label{Sect:rule-match}

We will now give a method to check the validity of equations in the
algebra $\Algbis{G}$ based on a variation on the bisimulation
technique. Our approach has strong similarities with, and is a
sharpening of, {\em FH-bisimulation}, as proposed by de Simone in
\cite{dS84,dS85}. (We remark, in passing, that FH-bisimilarity checking
has been implemented in the tool ECRINS~\cite{DoumencMS,MV91}.)
%%\lyammer{To do: Perhaps discuss the following definition informally.}

\begin{definition}[Rule-matching Bisimulation]
\label{open-bis-defn}
Let $G$ be a GSOS language. A relation $\approx {\subseteq}
\terms{\Sigma_G} \times \terms{\Sigma_G}$ is a {\em rule-matching
bisimulation} if it is symmetric and $ P \approx Q$ implies
\begin{description}
\item for each ruloid $\proofrule{H}{P\mv{a}P'}$ in the ruloid set of $P$, 
there exists a finite set $J$ of valid ruloids for $Q$ such that:
\begin{enumerate}
\item \label{obisim-1}
      For every {$\rho' = \frac{H'}{Q \mv{a'} Q' }\in J$}, we have:
      \begin{enumerate}
      \item {$a'=a$}, \label{obisim-1-1}
      \item {$P' \approx Q'$}, \label{obisim-1-2}
      \item {$\parens{\targetvars{\rho'} \cup \targetvars{\rho}} \cap
            \parens{\sourcevars{\rho} \cup \sourcevars{\rho'}} =
            \emptyset$} and
            \label{obisim-1-3}
      \item if {$y \in \targetvars{\rho} \intersect
            \targetvars{\rho'}$}, then {$x \mv{b} y \in H \intersect
            H'$} for some source variable {$x \in \sourcevars{\rho}
            \intersect \sourcevars{\rho'}$} and action {$b$}.
            \label{obisim-1-4}
      \end{enumerate}
\item {$\models_G \hyps{\rho} \imp\hyps{J}$}. \label{obisim-2}
\end{enumerate}
\end{description} 
We write $P \obisim_G Q$ if there exists a rule-matching bisimulation
$\approx$ relating $P$ and $Q$. We sometimes refer to the relation
$\obisim_G$ as {\em rule-matching bisimilarity}.
\end{definition} 
Note that, as the source and target variables of GSOS rules and
ruloids are distinct, condition~\ref{obisim-1-3} is equivalent to
{$\targetvars{\rho} \intersect \sourcevars{\rho'} = \emptyset$} and
{$\targetvars{\rho'} \intersect \sourcevars{\rho} =
\emptyset$}. Moreover,  {$\obisim_G$} is just
standard bisimilarity over closed terms.

%\lyammer{Point out that, for closed terms, {$\obisim_G$} is just standard
%bisimulation.}

Of course, the notion of rule-matching bisimulation is reasonable only
if we can prove that it is sound with respect to the standard
extension of bisimulation equivalence to open terms. This is the
import of the following theorem. 
%%, whose proof is in Appendix~\ref{appendixs1}.

\begin{theorem}[Soundness]\label{soundness}  
Let $G$ be a GSOS language. Then, for all $P,Q \in\terms{\Sigma_G}$, $P \obisim_G
Q$ implies $\Algbis{G} \models P = Q$.
\end{theorem}
The import of the above theorem is that, when trying to establish the
equivalence of two contexts $P$ and $Q$ in a GSOS language $G$, it is
sufficient to exhibit a rule-matching bisimulation relating them. A
natural question to ask is whether the notion of rule-matching
bisimulation is {\em complete} with respect to equality in
$\Algbis{G}$, {\ie} whether $\Algbis{G} \sat P=Q$ {implies} $P \obisim
Q$, for all $P,Q\in \terms{\Sigma_G}$.  
Below, we shall provide a counter-example to the above statement.

\begin{example}\label{Ex:alldisjointexts}
Consider a GSOS language $G$ consisting of a constant 
$(a+b)^\omega$ with rules

\TwoRules {} {(a+b)^\omega \mv{a} (a+b)^\omega}
	  {}{(a+b)^\omega \mv{b} (a+b)^\omega}

\noindent 
and unary function symbols $f$, $g$, $h$ and $i$ with rules
\FourRulesFullStop{x\mv{a}y_1, x\mv{b}y_2} {h(x) \mv{a} f(x)}
	  {x\mv{a}y_1, x\mv{b}y_2}{i(x) \mv{a} g(x)} 
          {x\mv{a}y_1, x\mv{b}y_2} { f(x) \mv{a} f(x)}
	  {x\mv{a}y_1}{g(x) \mv{a} g(x)}

\noindent 
First of all, note that no rule in $G$ is junk as the hypotheses of
each of the above rules are satisfiable. 

We claim that $\Algbis{G} \sat h(x)=i(x)$. To see this, it
is sufficient to note that, for all $p\in\closed{\Sigma_G}$, 
\begin{eqnarray*}
h(p) \mv{c} r & \Iff & p\mv{a},~p\mv{b},~ c=a \mbox{ and } r\syneq f(p) 
\quad \text{and} \\
i(p) \mv{c} r & \Iff & p\mv{a},~p\mv{b},~ c=a \mbox{ and } r\syneq g(p)
\enspace .
\end{eqnarray*}
Moreover, for a term $p$ such that $a,b\in\init{p}$, it is immediate to see 
that $f(p) \bisim g(p)$ as both these terms can only perform action $a$ 
indefinitely. 

However, $h(x)$ and $i(x)$ are {\em not} rule-matching bisimilar. In
fact, in order for $h(x) \obisim i(x)$ to hold, it must be the case
that $f(x) \obisim g(x)$.  This does not hold as the unique rule for
$g(x)$ cannot be matched by the rule for $f(x)$ because $\not\models
(x\mv{a}) \Rightarrow (x\mv{a} \mathbin{\wedge}~x\mv{b})$. Take,
{\eg}, a closed substitution $\sigma$ such that $\sigma(x)\syneq
h((a+b)^\omega)$.

Intuitively, the failure of rule-matching bisimulation in the above
example is due to the fact that, in order for $\Algbis{G} \sat
h(x)=i(x)$ to hold, it is sufficient that $f(p)$ and $g(p)$ be
bisimilar for those terms $p$ which enable transitions from $h(p)$ and
$i(p)$, rather than for arbitrary instantiations.
\end{example}
Note that the equation discussed in
Example~\ref{Ex:alldisjointexts} is valid in each disjoint extension
of the GSOS language considered there. 
%%%%%%%%%%%%%%%%%%
In the following section we will provide examples that will,
hopefully, convince our readers that rule-matching bisimulation is a
tool which, albeit not complete, can be used to check the validity of
many interesting equations.

It is natural to ask oneself at this point whether rule-matching
bisimilarity is preserved by taking disjoint extensions, \ie, whether
an equation that has been proven to hold in a language $G$ using
rule-matching bisimilarity remains sound for each disjoint extension
of $G$. The following example shows that this is not the case. 
\begin{example}
Consider a GSOS language $G$ consisting of a constant $a^\omega$ with
rule $a^\omega \mv{a} a^\omega$ and unary operations $f$ and $g$ with
the following rules. 

\TwoRules {x\mv{a}x'} {f(x) \mv{a} f(x)}
	  {y\mv{a}y'}{g(y) \mv{a} g(y)} 

\noindent 
First of all, note that no rule in $G$ is junk as the hypotheses of
each of the above rules are satisfiable. 
%%Hence none of the above rules
%%is redundant.

We claim that $\Algbis{G} \sat f(x)=g(y)$. To see this, it is
sufficient to note that each closed term in the language is bisimilar
to $a^\omega$. Moreover, $f(x) \obisim_G g(y)$ holds because the
formulae $x \mv{a}$ and $y \mv{a}$ are logically equivalent in $G$. On
the other hand, consider the disjoint extension $G'$ of $G$ obtained by
adding the constant $\nil$ with no rules to $G$. In this disjoint
extension, $f(x) \obisim_{G'} g(y)$ does {\em not} hold because $x
\mv{a}$ does not entail $y \mv{a}$.
\end{example}
However, rule-matching bisimilarity in language $G$ {\em is} preserved by
taking disjoint extensions if the language $G$ is sufficiently
expressive in the sense formalized by the following result.
\begin{theorem}\label{soundnessdis}  
Let $G$ be a GSOS language such that
$\init{\closed{\Sigma_G}}=2^{\Act}$. Then, for all $P,Q
\in\terms{\Sigma_G}$, $P \obisim_G Q$ implies $\Classbis{G} \models P
= Q$.
\end{theorem}
\begin{proof}
The proof of Theorem~\ref{soundness} can be replayed, making use of
the observations that for each disjoint extension $G'$ of $G$, the
collection of ruloids in $G'$ for a $\Sigma_G$-term $P$ coincides with
the collection of ruloids for $P$ in $G$. Moreover, in light of the
proviso of the theorem, $\sat_G F \Implies F'$ iff $\sat_{G'} F \Implies
F'$, for all formulae $F$ and $F'$.
\end{proof}

\noindent
A conceptually interesting consequence of the above result is that,
when applied to a sufficiently expressive GSOS language $G$,
rule-matching bisimilarity is a proof method that is, in some sense,
{\em monotonic with respect to taking disjoint extensions of the
original language}. This means that rule-matching bisimilarity can
only prove the validity of equations in $G$ that remain true in all
its disjoint extensions. A similar limitation applies to the proof
methods presented in, \eg,~\cite{dS85,Weerdenburg2008}.

\section{Examples}\label{Sect:examples}

We shall now present some examples of applications of the
`rule-matching bisimulation technique'. In particular, we shall show
how some well known equations found in the literature on process
algebra can be verified using it. 
%%We discuss two other examples in Appendix~\ref{Sect:expapp}.

\paragraph{Associativity of Sequencing}
Let $G$ be any GSOS language containing the sequencing operation specified by 
(\ref{rules4sequencing}). Let {$R[x,y,z]= x;(y;z)$} and {$L[x,y,z]=(x;y);z$}.
The ruloids for these two contexts were given in (\ref{ruloids4sequencing}).

Consider the symmetric closure of the relation 
\begin{eqnarray*}
\approx & \deq & \set{ (R[x,y,z],L[x,y,z]) \mid x,y,z\in\Var} \cup \cali 
\end{eqnarray*}
where $\cali$ denotes the identity relation over
$\terms{\Sigma_G}$. By Theorem~\ref{soundness}, to show that the
contexts $L$ and $R$ are equivalent, it is sufficient to check that
what we have just defined is a rule-matching bisimulation. In
particular, we need to check the correspondence between the ruloids
for these contexts (which is the one given in
(\ref{ruloids4sequencing})), and then check that the targets are
related by $\approx$. The verification of these facts is trivial. Thus
we have shown that sequencing is associative in any GSOS language that
contains the sequencing operation.

The associativity proofs for the standard parallel composition
operators found in {\eg}~ACP, CCS, SCCS and {\sc Meije}, and for the
choice operators in those calculi follow similar lines.  

%%% Keep it for the full version? 
%%% \lyammer{We could give an example here.}

\paragraph{Commutativity of Interleaving Parallel Composition}

%Let's do {$\interleave$} for fun. 
Many standard axiomatizations of behavioural equivalences in the
literature cannot be used to show that, {\eg}, parallel composition is
commutative and associative. We will now show how this can be easily
done using the rule-matching bisimulation technique. We will exemplify
the methods by showing that the interleaving parallel composition
operation $\interleave$ \cite{Ho85} is commutative.

We recall that the rules for $\interleave$ are (one pair of rules for
each $a\in\Act$): 
\TwoRulesLabFullStop {x \mv{a} x'} {x\interleave y \mv{a}
x'\interleave y} {y \mv{a} y'} {x\interleave y \mv{a} x\interleave
y'}{parallel-rules} 
The ruloids for the contexts $x\interleave y$ and
$y\interleave x$ given by Theorem~\ref{ruloid-thm} are the following ones (one pair of
ruloids for each $a\in\Act$). \TwoRules {x \mv{a} x'} {x\interleave
y \mv{a} x'\interleave y \\ y \interleave x \mv{a} y \interleave x' }
{y \mv{a} y'} {x\interleave y \mv{a} x\interleave y' \\ y \interleave
x \mv{a} y' \interleave x } It is now immediate to
see that the relation $\set{ (x\interleave y,y\interleave x) \mid
x,y\in\Var}$ is a rule-matching bisimulation in any GSOS language that
includes the interleaving operator. In fact, the correspondence
between the ruloids is trivial and the targets are related by the
above relation.

\paragraph{De Simone's Clock Example}
\label{clock}
In his seminal paper \cite{dS85}, de Simone presents a bisimulation
based technique useful for proving open equations between contexts
specified using the so-called de Simone format of operational
rules. On page 260 of that paper, de Simone discusses two examples
showing that there are valid open equalities between contexts that his
technique cannot handle.  Below, we shall discuss a variation on one
of his examples, the {\em clock example}, which maintains all the
characteristics of the original one in \cite{dS85}, showing how
rule-matching bisimulations can be used to check the relevant
equalities.
\begin{figure}[t]
\begin{quote}
Fix a partial, commutative and associative
function $\gamma : \Act\times\Act\rightharpoonup\Act$, which describes
the synchronization between actions. 
The $\|$ operation can be described by the
rules (for all $a, b, c \in\Act$):
\[
\proofrule{x \mv{a} x'}{x \| y \mv{a} x' \| y} \qquad\qquad
\proofrule{y \mv{a} y'}{x \| y \mv{a} x \| y'} \qquad\qquad
\proofrule{x \mv{a} x' ,~ y \mv{b} y'}{x \| y \mv{c} x' \| y'}
 ~ \gamma(a,b) = c
\]
\end{quote}
\caption{The rules for $\|$\label{par+sync}}
\end{figure}

Suppose we have a GSOS language which includes parallel composition
with synchronization, $\|$, described by the rules in
Figure~\ref{par+sync}, the interleaving operation, $\interleave$,
described by the rules \eqnref{parallel-rules}, and a constant
$\clock$ (the {\em clock over the whole set of actions} in de Simone's
terminology) with rules
$$\clock \mv{a} \clock ~~~ (a\in\Act) \enspace .$$
Consider the contexts $C[x]\syneq x \| \clock$ and $D[x]\syneq x
\interleave \clock$. We do have that, regardless of the precise
description of $G$, the terms $C[x]$, $D[x]$ and $\clock$ are all equal in
$\Algbis{G}$. This can be easily shown by establishing that the
symmetric closures of the relations $\set{ (C[p],\clock) \mid p\in
\closed{\Sigma_G}}$ and $\set{ (D[p],\clock) \mid p\in
\closed{\Sigma_G}}$ are bisimulations. However, as argued in
\cite{dS85}, de Simone's techniques based on FH-bisimilarity cannot be
used to establish these equalities. We can instead show their
validity using our rule-matching bisimulation technique as follows.

First of all, we compute the ruloids for the contexts $C[x]$ and $D[x]$. 
These are, respectively,  
\[
\proofrule{}{C[x] \mv{a} C[x]} ~ (a\in\Act) \qquad
\proofrule{x \mv{a} x'}{C[x] \mv{a} C[x']} ~ (a\in\Act) \qquad
\proofrule{x \mv{a} x' }{C[x] \mv{b} C[x'] }
 ~ \exists c\in\Act: \gamma(a,c) =b 
  \]
and 
\[
\proofrule{}{D[x] \mv{a} D[x]} ~ (a\in\Act) \qquad
\proofrule{x \mv{a} x'}{D[x] \mv{a} D[x']} ~ (a\in\Act) \enspace .
\]
Now, it can be easily checked that the symmetric closure of the relation 
$$
\set{ (C[x],D[z]), (C[x],\clock), (D[x],\clock) \mid x,z\in\Var}
$$
is a rule-matching bisimulation. The point is that {\em any} ruloid for $C[x]$ can 
be matched by an axiom for $D[z]$, and, vice versa, {\em any} ruloid for 
$D[z]$ can be matched by an axiom for $C[x]$. This is because it is always 
the case that $\models (x \mv{a}) \Implies \T$ for $x\in\Var$.

\section{Partial Completeness Results}\label{Sect:partial-completeness}

%%\lyammer{Is there a reasonable class of systems for which the above method 
%% is complete? We might try to exhibit some partial completeness results in
%%this section, if we come up with any.}

In previous sections, we showed that the rule-matching bisimulation
technique, albeit not complete in general, can be used to prove
several important equations found in the literature on process
algebras. In particular, the soundness of all the equations generated
by the methods in \cite{ABV94} can be proven by exhibiting appropriate
rule-matching bisimulations.
%%\lyammer{We need to check this more carefully.}  
A natural question to ask is whether there are some classes of
contexts for which rule-matching bisimulations give us a complete proof
technique for establishing equality between contexts. One such class
of contexts is, of course, that of closed terms, as rule-matching bisimilarity
coincides with bisimilarity over processes.

Below we will present another partial completeness result, this time
with respect to a class of contexts that we call `persistent'.

\begin{definition}
Let $G$ be a GSOS language and $P\in\terms{\Sigma_G}$. We say that $P$
is {\em persistent} iff each ruloid in $\ruloids{P}$ is of the form
$\frac{H}{P \mv{a} P}$
%%\Rule{H}{P \mv{a} P}
for some $a\in\Act$.
\end{definition}
Thus persistent contexts are terms that test their arguments, perform actions
according to the results of these tests, and then remain unchanged. 
\begin{theorem}[Completeness for Persistent Contexts]
\label{persistent-completeness}
Let $G$ be a GSOS language. Then $\Algbis{G} \models P = Q$ iff $P
\obisim Q$, for all persistent $P,Q\in\terms{\Sigma_G}$.
\end{theorem}
We now proceed to introduce another class of operations for which
rule-matching bisimilarity yields a complete proof method. 

\begin{definition}[Non-inheriting Rule]\label{Def:non-inherit}
A GSOS rule of the form (\ref{general-gsos-rule}) is {\em
non-inheriting} if none of the variables in $\xvec$, namely the
source variables in the rule, occurs in the target of the conclusion
of the rule $C[\vec{x}, \vec{y}]$. A GSOS language is {\em
non-inheriting} if so is each of its rules.  Non-inheriting de Simone
rules and languages are defined similarly.
\end{definition}

\begin{theorem}
\label{thm:gsos-noninheriting}
Let $G$ be a non-inheriting GSOS language that, for each
$P\in\terms{\Sigma_G}$ and $c\in\Act$, contains at most one ruloid for
$P$ having $c\in\Act$ as action. Let $G'$ be the disjoint extension of
$G$ obtained by adding to $G$ the operations and rules of the language
BCCSP~\cite{vG2001,Mi89} with $\Act$ as set of actions.  Let $P$ and
$Q$ be terms over $\Sigma_G$. Then 
$\Algbis{G'} \models P = Q$ implies $P \obisim_{G'} Q$. 
\end{theorem}
A minor modification of the proof for the above result
yields a partial completeness result for a class of de
Simone systems.
\begin{theorem}
\label{thm:desimone-noninheriting}
Let $G$ be a non-inheriting de Simone language that, for each
$f\in\Sigma_G$ and $c\in\Act$, contains at most one rule having
$f\in\Sigma_G$ as principal operation and $c\in\Act$ as action. Let
$G'$ be the disjoint extension of $G$ obtained by adding to $G$ the
constant $\nil$ and the $\Act$-labelled prefixing operations from the
language BCCSP~\cite{vG2001,Mi89}.  Let $P$ and $Q$ be terms over
$\Sigma_G$. Then $\Algbis{G'} \models P = Q$ implies $P \obisim_{G'}
Q$.
\end{theorem}
For instance, the above theorem yields that rule-matching bisimilarity
can prove all the sound equations between terms constructed using
variables and the operations of restriction and injective relabelling
from CCS~\cite{Mi89} and synchronous parallel composition from
CSP~\cite{Ho85}.

\section{Related and Future Work}\label{Sect:related-work}

The development of general methods for proving equivalences between
open terms in expressive process calculi is a challenging subject that
has received some attention since the early developments of the
algebraic theory of processes---see, \eg, the
references~\cite{BruniFMM00,LarsenX1991,Rensink00,dS85,Weerdenburg2008}
for some of the work in this area.  De Simone's
FH-bisimilarity~\cite{dS85} represents an early meaningful step
towards a general account of the problem, presenting for the first time a
sound bisimulation method in place of the usual definition which
involves the closure under all possible substitutions. Our method
relies mainly on the concepts underlying FH-bisimilarity and it is a
refinement of that notion in the more expressive setting of GSOS
languages. (See de Simone's `Clock Example' discussed on
page~\pageref{clock}, where FH-bisimilarity fails while $\obisim$
succeeds.)

Later Rensink addressed the problem of checking bisimilarity of open
terms in~\cite{Rensink00}, where he presented a natural sharpening of
de Simone's FH-bisimilarity.  His extension of FH-bisimilarity is
orthogonal to ours and provides another method to check equivalences
between open terms that is more powerful than the original
FH-bisimilarity. Rensink defined a new notion of
bisimulation equivalence, called {\em hypothesis preserving bisimilarity}, that
adds to FH-bisimilarity the capability to store some kind of
information about the variable transitions during the
computation.  

To explain the import of hypothesis preserving bisimilarity we can
look at Example~\ref{Ex:alldisjointexts}. We note that $\obisim$ fails
to establish the sound equation $h(x)=i(x)$ because at the second step
of the computation some knowledge about the transitions of the closed
term $p$ substituted for $x$ is already established (indeed, at that
point we know that $p$ performs a $b$-transition, since this has been
tested at the first step). Nevertheless, when comparing $f(x)$ and
$g(x)$, rule-matching bisimulation behaves in memoryless fashion and
ignores this information.  Rensink's hypothesis preserving
bisimilarity takes into account the history and this is enough to
overcome the difficulties in that example and analogous scenarios.
Adding this feature to $\obisim$ would lead to a more powerful
rule-matching equivalence; we leave this further sharpening for future
work together with extensions of $\obisim$ to more expressive rule
formats.

Recently, van Weerdenburg addressed the automation of soundness proofs
in~\cite{Weerdenburg2008}. His approach differs from the one
in~\cite{Rensink00,dS85} and ours since he translates the operational
semantics into a logical framework. In such a framework, rules are
encoded as logical formulae and the overall semantics turns out to be
a logical theory, for which van Weerdenburg provides a sequent
calculus style proof system.  In the aforementioned paper, he offers
some examples of equivalences from the literature that can be proved
using his method in order to highlight its applicability. However,
even though the ultimate aim of the research described
in~\cite{Weerdenburg2008} is the automation of soundness proofs, van
Weerdenburg's system presents some drawbacks.  The main point is that
the user is not only required to provide the operational semantics and
the equation to check (together with the standard encoding of
bisimilarity), but he must also provide a candidate bisimulation
relation that can be used to show the validity of the equation under
consideration together with all the axioms that are needed to complete
the proof. The user is supposed thus to have a clear understanding of
what the proof is going to look like. This seems to be a general and
inescapable drawback when approaching the problem of checking
equations through a translation into a logical system. 

Despite the aforementioned slight drawback, the approach proposed by
van Weerdenburg is, however, very interesting and complements the
proposals that are based on the ideas underlying de Simone's
FH-bisimilarity, including ours. We believe that an adequate solution
to the problem of automating checks for the validity of equations in
process calculi will be based on a combination of bisimulation-based
and logical approaches.

A related line of work is the one pursued in, \eg, the
papers~\cite{AcetoBIMR2009,CranenMR08,MousaviIPL05}. Those papers
present rule formats that guarantee the soundness of certain algebraic
laws over a process language `by design', provided that the SOS rules
giving the semantics of certain operators fit that format. This is an
orthogonal line of investigation to the one reported in this
article. As a test case for the applicability of our rule-based
bisimilarity, we have checked that the soundness of all the equations
guaranteed to hold by the commutativity format
from~\cite{MousaviIPL05} can be shown using $\obisim$. We are carrying
out similar investigations for the rule formats proposed
in~\cite{AcetoBIMR2009,CranenMR08}.

Another avenue for future research we are actively pursuing is the
search for more, and more general, examples of partial completeness
results for rule-matching bisimulation over GSOS and de Simone
languages. Indeed, the partial completeness results we present in
Section~\ref{Sect:partial-completeness} are just preliminary steps
that leave substantial room for improvement. Last, but not least, we
are about to start working on an implementation of a prototype checker
for rule-matching bisimilarity.

\end{document}